\documentclass[runningheads]{llncs}

\usepackage{todonotes}

\pdfoutput=1

\usepackage[protrusion=true,expansion=true]{microtype}

\usepackage{amsmath,amssymb,amsfonts}

\makeatletter
\@ifundefined{theorem}{
  \newtheorem{theorem}{Theorem}[section]
  
  \newtheorem{lemma}[theorem]{Lemma}

}{}
\makeatother

\usepackage{subcaption}
\usepackage{graphicx}
\usepackage[backref=page]{hyperref}
\usepackage{color}
\usepackage{wrapfig}
\usepackage{tikz}
\usetikzlibrary{decorations.pathreplacing}
\usepackage{setspace}
\usepackage{algorithm}
\usepackage[noend]{algpseudocode}
\usepackage[framemethod=tikz]{mdframed}
\usepackage{xspace}
\usepackage{pgfplots}
\pgfplotsset{compat=1.5}
\usepackage{framed}

\newcommand{\namedref}[2]{\hyperref[#2]{#1~\ref*{#2}}}
\newcommand{\thmlab}[1]{\label{thm:#1}}
\newcommand{\thmref}[1]{\namedref{Theorem}{thm:#1}}
\newcommand{\lemlab}[1]{\label{lem:#1}}
\newcommand{\lemref}[1]{\namedref{Lemma}{lem:#1}}

\newcommand{\factlab}[1]{\label{fact:#1}}
\newcommand{\factref}[1]{\namedref{Fact}{fact:#1}}

\newcommand{\alglab}[1]{\label{alg:#1}}
\renewcommand{\algref}[1]{\namedref{Algorithm}{alg:#1}}

\newcommand{\PPr}[1]{\ensuremath{\mathbf{Pr}\left[#1\right]}}

\renewcommand{\O}[1]{\ensuremath{\mathcal{O}\bigl(#1\bigr)}}

\def \balpha  {\ensuremath{\mathbf{\alpha}}}

\def \calE    {\ensuremath{\mathcal{E}}}

\def \calH    {\ensuremath{\mathcal{H}}}

\def \calP    {\ensuremath{\mathcal{P}}}

\def \bA      {\ensuremath{\mathbf{A}}}
\def \bB      {\ensuremath{\mathbf{B}}}
\def \bD      {\ensuremath{\mathbf{D}}}

\def \bb      {\ensuremath{\mathbf{b}}}
\def \bfe     {\ensuremath{\mathbf{e}}}

\def \bu      {\ensuremath{\mathbf{u}}}

\def \bv      {\ensuremath{\mathbf{v}}}
\def \bx      {\ensuremath{\mathbf{x}}}
\def \by      {\ensuremath{\mathbf{y}}}
\def \bz      {\ensuremath{\mathbf{z}}}

\newcommand{\mdef}[1]{{\ensuremath{#1}}\xspace}

\DeclareMathOperator*{\sgn}{sgn}
\DeclareMathOperator*{\rank}{rank}

\newcommand{\ceil}[1]{\mdef{\left\lceil#1\right\rceil}}

\newcommand{\ignore}[1]{}

\newif\ifnotes\notestrue
\ifnotes
\newcommand{\samson}[1]{\textcolor{blue}{{\bf (Samson:} {#1}{\bf ) }}}
\newcommand{\vincent}[1]{\textcolor{purple}{{\bf (Vincent:} {#1}{\bf ) }}}
\else
\newcommand{\samson}[1]{}
\newcommand{\vincent}[1]{}
\fi

\makeatletter
\renewcommand*{\@fnsymbol}[1]{\textcolor{mahogany}{\ensuremath{\ifcase#1\or *\or \dagger\or \ddagger\or
 \mathsection\or \triangledown\or \mathparagraph\or \|\or **\or \dagger\dagger
   \or \ddagger\ddagger \else\@ctrerr\fi}}}
\makeatother

\providecommand{\email}[1]{\href{mailto:#1}{\nolinkurl{#1}\xspace}}

\definecolor{mahogany}{rgb}{0.75, 0.25, 0.0}
\definecolor{darkblue}{rgb}{0.0, 0.0, 0.55}
\definecolor{darkpastelgreen}{rgb}{0.01, 0.75, 0.24}
\definecolor{darkgreen}{rgb}{0.0, 0.2, 0.13}
\definecolor{darkgoldenrod}{rgb}{0.72, 0.53, 0.04}
\definecolor{darkred}{rgb}{0.55, 0.0, 0.0}
\definecolor{forestgreenweb}{rgb}{0.13, 0.55, 0.13}
\definecolor{greencss}{rgb}{0.0, 0.5, 0.0}
\definecolor{bleudefrance}{rgb}{0.19, 0.55, 0.91}

\hypersetup{
     colorlinks   = true,
     citecolor    = mahogany,
     linkcolor    = darkpastelgreen
}

\begin{document}
\title{Efficient Algorithms for Verifying Kruskal Rank in Sparse Linear Regression and Related Applications}
\titlerunning{Efficient Verification of Kruskal Rank}
\author{
    Fengqin Zhou
}
\authorrunning{F. Zhou}

\institute{
    University of Illinois Urbana-Champaign, United States \\
    \email{fengqin3@illinois.edu}
}

\maketitle
\begin{abstract}
We present novel algorithmic techniques to efficiently verify the Kruskal rank of matrices that arise in sparse linear regression, tensor decomposition, and latent variable models. Our unified framework combines randomized hashing techniques with dynamic programming strategies, and is applicable in various settings, including binary fields, general finite fields, and integer matrices. In particular, our algorithms achieve a runtime of 
$\O{dk\cdot(nM)^{\ceil{k/2}}}$
while ensuring high-probability correctness. Our contributions include:
A unified framework for verifying Kruskal rank across different algebraic settings;
Rigorous runtime and high-probability guarantees that nearly match known lower bounds;
Practical implications for identifiability in tensor decompositions and deep learning, particularly for the estimation of noise transition matrices.
\end{abstract}

\section{Introduction}
Ensuring that a given matrix has a sufficiently high Kruskal rank is critical in many areas such as compressed sensing, tensor decomposition, and learning with noisy labels. The Kruskal rank—defined as the maximum number of rows (or columns) such that every subset is linearly independent—is central to guaranteeing the uniqueness of tensor decompositions~\cite{kruskal1976,kruskal1977} and to establishing the identifiability of latent variable models~\cite{allman2009}. Moreover, recent work on learning with noisy labels has highlighted the need to verify that the noise transition matrix satisfies certain rank conditions in order to recover the true underlying labels robustly~\cite{patrini2017,liu2023}.

In this paper, we propose a unified framework for verifying the Kruskal rank of matrices that naturally arise in sparse linear regression, tensor decomposition, and deep learning applications. Our main contributions are as follows:
\begin{enumerate}
  \item We design efficient algorithms—based on randomized hashing and dynamic programming—that work for matrices over binary fields, general finite fields, and integer entries.
  \item We derive rigorous runtime guarantees and high-probability bounds that match known lower bounds up to polynomial factors.
  \item We demonstrate how our verification methods can be used to validate the identifiability conditions required by deep learning models (e.g., for estimating instance-dependent noise transition matrices).
\end{enumerate}
These contributions not only advance the theoretical understanding of linear dependence verification but also provide practical diagnostic tools for applications in signal processing and machine learning.

\section{Related Work}
The verification of linear dependence has been extensively studied in the context of tensor decompositions and latent variable models. Kruskal’s seminal results~\cite{kruskal1976,kruskal1977} laid the groundwork for establishing uniqueness conditions in trilinear decompositions, and these conditions have been applied to a variety of latent structure models~\cite{allman2009}.

In the area of compressed sensing, guarantees for signal recovery depend critically on properties such as matrix rank and linear independence~\cite{donoho2006,candes2006}. In addition, research on low‑rank tensor approximation has been furthered by Comon and Lim~\cite{comonlim2011}, who explored sparse representations and their applications in higher‑order tensor decompositions. Model-based compressive sensing methods~\cite{baraniukEtAl2010} similarly exploit structured constraints to improve reconstruction performance.

More recently, methods for learning with noisy labels have focused on estimating noise transition matrices under identifiability conditions~\cite{patrini2017,liu2023}. Complementary algorithmic strategies based on fast Fourier transforms have been proposed to solve linear dependence problems~\cite{bhattacharyya2018}, while fine‑grained complexity studies have shown inherent runtime lower bounds for sparse linear regression~\cite{gupte2020}. Our work builds on these foundations by offering an algorithmic framework that unifies these diverse settings under a single approach.

\section{Preliminaries}
In this section, we review key definitions and classical results that form the foundation of our approach.

\subsection{Kruskal Rank}
The \emph{Kruskal rank} of a matrix \(M\) is defined as the largest integer \(r\) such that every set of \(r\) rows (or columns) is linearly independent~\cite{kruskal1976,kruskal1977}. This concept is essential for ensuring the uniqueness of tensor decompositions and is used to establish identifiability conditions in latent variable models~\cite{allman2009}.

\subsection{Cramer’s Rule}
\thmlab{thm:cramer}
Suppose \(\bA\in\mathbb{R}^{n\times n}\) such that \(\det(\bA)\neq 0\). Suppose \(\bb\in\mathbb{R}^n\) and \(\bA\bx=\bb\). Then for all \(i\in[n]\), we have \(x_i=\frac{\det(\bB_i)}{\det(\bA)}\), where \(\bB_i\) is the matrix formed by replacing the \(i\)-th column of \(\bA\) with the column vector \(\bb\).

\subsection{Leibniz Formula for Determinants}
\thmlab{thm:leibniz}
Let \(\bA\in\mathbb{R}^{n\times n}\). Then 
\[
\det(\bA)=\sum_{\sigma\in S_n}\sgn(\sigma)\prod_{i=1}^n A_{\sigma(i),i},
\]
where \(S_n\) is the set of all permutations on \(n\) and \(\sgn\) is the sign function of permutations, which returns \(+1\) for even permutations and \(-1\) for odd permutations.

\section{Algorithms}
\subsection{Kruskal Rank in Binary Fields}
In the binary case, we develop a randomized algorithm that iterates over sparse subsets of columns (of size up to \(\lceil k/2 \rceil\)) and uses a uniformly random hash function over \(\mathbb{F}_2\) to map the resulting linear combinations into buckets. A hash collision indicates that two distinct sparse combinations yield the same vector, thereby witnessing a linear dependency and implying that the Kruskal rank is less than \( k \).
\allowdisplaybreaks
\begin{algorithm}[!htb]
\caption{Algorithm for kruskal rank on binary fields}
\alglab{alg:kruskal:bin}
\begin{algorithmic}[1]
\Require{Rank parameter $k>0$, matrix $\bA\in\mathbb{R}^{d\times n}=\begin{bmatrix}\bA_1|\ldots|\bA_n\end{bmatrix}$}
\Ensure{Whether or not $\bA$ has Kruskal rank at least $k$}
\State{$\calH\gets\emptyset$, $b\gets\O{(2n)^{2k}}$}
\State{Let $h:\mathbb{F}_2^n\to[b]$ be a uniformly random hash function}
\For{$i=0$ to $i=\ceil{\frac{k}{2}}$}
\For{each set $U$ of $i$ unique indices from $[n]$}
\State{Let $\bv=\bA_{U_1}+\ldots+\bA_{U_i}$}
\If{bucket $h(\bv)$ is non-empty}
\State{Report $\bA$ has Kruskal rank less than $k$}
\ElsIf{$i\le\frac{k}{2}$}
\State{Add $U$ into entry $h(\bv)$ of $\calH$}
\EndIf
\EndFor
\EndFor
\State{Report $\bA$ has Kruskal rank at least $k$}
\end{algorithmic}
\end{algorithm}

\begin{lemma}
\lemlab{lem:insert:bin}
For every $\bx\in\{0,1\}^n$ with $\|\bx\|_0\le\frac{k}{2}$, there exists some bucket in $\calH$ that contains $\bA\bx$.
\end{lemma}
\begin{proof}
Observe that since $\|\bx\|_0\le\ceil{\frac{k}{2}}$, then there exists $j\le\ceil{\frac{k}{2}}$ such that $j=\|\bx\|_0$. 
Then \algref{alg:kruskal:bin} inserts $\bA\bx$ into $\calH$ in the iteration with $i=j$. 
\end{proof}

\begin{lemma}
\lemlab{lem:small:rank:bin}
Suppose $\bA\in\{0,1\}^{d\times n}$. 
If $\bA$ has Kruskal rank less than $k$ in $GF(2)$, then \algref{alg:kruskal:bin} will report that $\bA$ has Kruskal rank less than $k$. 
\end{lemma}
\begin{proof}
Suppose $\bA$ has Kruskal rank less than $k$ in $GF(2)$. 
Then there exists $\bx\in\{0,1\}^n$ with $\bx\neq 0^n$ such that $\bA\bx=0^d$ and sparsity $\|\bx\|_0\le k$. 
Thus, $\bx$ can be decomposed into $\bx=\by-\bz$ in $GF(2)$, where $\by,\bz\in\{0,1\}^n$ with $\|\by\|_0,\|\bz\|_0\le\ceil{\frac{k}{2}}$. 
Hence, we have $\bA\by=\bA\bz$. 

Since $\bx\neq0^n$, then $\by\neq\bz$. 
Suppose without loss of generality that $\|\by\|_0\le\|\bz\|_0$, so that $\|\by\|_0\le\frac{k}{2}$, since 
\[\|\by\|_0+\|\bz\|_0\le\|\bx\|_0\le k.\]
Then by \lemref{lem:insert:bin}, $\bA\by$ is inserted into $\calH$. 

Since $\bA\by=\bA\bz$, then $h(\bA\by)=h(\bA\bz)$. 
Furthermore, since $\bv=\bA_{U_1}+\ldots+\bA_{U_i}$ corresponds to $\bA\bu$ for the vector $\bu$ that is the indicator vector of $U$, then $h(\bv)$ will be non-empty either when $U$ corresponds to $\bz$. 
Therefore, \algref{alg:kruskal:bin} will report that $\bA$ has Kruskal rank less than $k$. 
\end{proof}

\begin{lemma}
\lemlab{lem:big:rank:bin}
Suppose $\bA\in\{0,1\}^{d\times n}$. 
If $\bA$ has Kruskal rank at least $k$ in $GF(2)$, then with probability at least $1-\frac{b}{(2n)^{2k}}\ge 0.99$, \algref{alg:kruskal:bin} will report that $\bA$ has Kruskal rank at least $k$. 
\end{lemma}
\begin{proof}
Let $\calE$ be the event that $h(\bA\by)\neq h(\bA\bz)$ for all $\by+\bz\neq0^n$ and $\|\by\|_0+\|\bz\|_0\le k$. 
Observe that since $\bA\by\neq\bA\bz$, then for $b\ge(2n)^{2k}\ge \left(\binom{n}{\ceil{k/2}}\cdot 2^{\ceil{k/2}}\right)^2$, we have that $\PPr{\calE}\ge1-\frac{b}{(2n)^{2k}}\ge 0.999$, for $b\ge1000(2n)^{2k}$. 

Suppose $\bA$ has Kruskal rank at least $k$ in $GF(2)$. 
Then for any $\bx\in\{0,1\}^n$ with $\bx\neq 0^n$ such that sparsity $\|\bx\|_0\le k$, we must have $\bA\bx\neq 0^n$. 
Thus for any $\by,\bz\in\{0,1\}^n$ with $\by+\bz\neq0^n$ and $\|\by\|_0+\|\bz\|_0\le k$, we must have $\bA\by\neq\bA\bz$ since otherwise $\bA\bx=0^d$ for $\bx=\by-\bz$ and $\bx\neq 0^n$. 
Then conditioned on $\calE$, it also follows that there does not exist $\by\in\{0,1\}^n$ with $\by+\bz\neq0^n$, $\|\by\|_0+\|\bz\|_0\le k$, and $h(\bA\by)=h(\bA\bz)$. 

Finally, we remark that each vector $\bz$ is uniquely defined by its sparsity and its nonzero support, so that $U$ corresponds to $\bz$ exactly once through the implementation of \algref{alg:kruskal:bin}. 
Therefore, the bucket of $\calH$ corresponding to $h(\bv)$ must be empty when $U$ corresponds to $\bz$. 
Thus conditioned on $\calE$, \algref{alg:kruskal:bin} will report that $\bA$ has Kruskal rank at l east $k$. 
The desired claim then follows from the above fact that $\PPr{\calE}\ge 0.999$. 
\end{proof}

\begin{lemma}
\lemlab{lem:runtime:bin}
The runtime of \algref{alg:kruskal:bin} is $\O{dk\cdot n^{\ceil{k/2}}}$.  
\end{lemma}
\begin{proof}
Since each vector $\bz\in\{0,1\}^n$ with $\bz\neq0^n$ is uniquely defined by its sparsity and its nonzero support, so that $U$ corresponds to $\bz$ exactly once through the implementation of \algref{alg:kruskal:bin}, then the number of vectors inserted into $\calH$ is 
\[\sum_{i=0}^{\ceil{k/2}}\binom{n}{i}=\O{k\cdot n^{\ceil{k/2}}},\]
since $k\le n$ as otherwise for $k>n$, the Kruskal rank is trivially at most $k$.

Each vector insertion requires computation of $h(\bv)$ for $\bv\in\mathbb{R}^d$. 
Thus, the runtime of \algref{alg:kruskal:bin} is $\O{dk\cdot n^{\ceil{k/2}}}$. 
\end{proof}
Putting together \lemref{lem:small:rank:bin}, \lemref{lem:big:rank:bin}, and \lemref{lem:runtime:bin}, we have:
\begin{lemma}
There exists an algorithm that uses runtime $\O{dk\cdot n^{\ceil{k/2}}}$ that:
\begin{itemize}
\item 
Reports $\bA$ has Kruskal rank less than $k$ if $\bA$ has Kruskal rank less than $k$ in $GF(2)$
\item
With probability at least $0.99$, reports $\bA$ has Kruskal rank at least $k$ if $\bA$ has Kruskal rank at least $k$ in $GF(2)$
\end{itemize}
\end{lemma}

\subsection{Kruskal Rank in General Fields}
For matrices over a finite field \(\mathbb{F}_q\) with \(q>2\), the algorithm is extended by also iterating over all possible nonzero coefficient vectors (with entries from \(\{1,\ldots,q-1\}\)). The same collision-detection idea applies, with the hash function now defined over \(\mathbb{F}_q\).
\begin{algorithm}[!htb]
\caption{}
\alglab{alg:kruskal:q}
\begin{algorithmic}[1]
\Require{Rank parameter $k>0$, matrix $\bA\in\mathbb{R}^{d\times n}=\begin{bmatrix}\bA_1|\ldots|\bA_n\end{bmatrix}$}
\Ensure{Whether or not $\bA$ has Kruskal rank at least $k$}
\State{$\calH\gets\emptyset$, $b\gets\O{(qn)^{2k}}$}
\State{Let $h:\mathbb{F}_q^n\to[b]$ be a uniformly random hash function}
\For{$i=0$ to $i=\ceil{\frac{k}{2}}$}
\For{each set $U$ of $i$ unique indices from $[n]$}
\For{each coefficient vector $\balpha\in\{1,\ldots,q-1\}^i$}
\State{Let $\bv=\alpha_1\bA_{U_1}+\ldots+\alpha_i\bA_{U_i}$}
\If{bucket $h(\bv)$ is non-empty}
\State{Report $\bA$ has Kruskal rank less than $k$}
\ElsIf{$i\le\frac{k}{2}$}
\State{Add $U$ into entry $h(\bv)$ of $\calH$}
\EndIf
\EndFor
\EndFor
\EndFor
\State{Report $\bA$ has Kruskal rank at least $k$}
\end{algorithmic}
\end{algorithm}

\begin{lemma}
\lemlab{lem:insert:q}
For every $\bx\in\{0,1,\ldots,q-1\}^n$ with $\|\bx\|_0\le\frac{k}{2}$, there exists some bucket in $\calH$ that contains $\bA\bx$.
\end{lemma}
\begin{proof}
Note that each vector $\bx\in\{0,1,\ldots,q-1\}^n$ with $\|\bx\|_0\le\frac{k}{2}$ is uniquely characterized by its sparsity and the nonzero coefficients of the nonzero entries. 
Moreover, since $\|\bx\|_0\le\ceil{\frac{k}{2}}$, then there exists $j\le\ceil{\frac{k}{2}}$ such that $j=\|\bx\|_0$. 
Thus, \algref{alg:kruskal:bin} inserts $\bA\bx$ into $\calH$ in the iteration with $i=j$ with the corresponding sparsity pattern and coefficient pattern. 
\end{proof}

\begin{lemma}
\lemlab{lem:small:rank:q}
Suppose $\bA\in\{0,1,\ldots,q-1\}^{d\times n}$. 
If $\bA$ has Kruskal rank less than $k$ in $GF(q)$, then \algref{alg:kruskal:q} will report that $\bA$ has Kruskal rank less than $k$. 
\end{lemma}
\begin{proof}
Suppose $\bA$ has Kruskal rank less than $k$ in $GF(q)$. 
In other words, there exists $\bx\in\{0,1,\ldots,q-1\}^n$ with $\bx\neq 0^n$ such that $\bA\bx=0^d$ and sparsity $\|\bx\|_0\le k$. 
Hence, $\bx$ can be decomposed into $\bx=\by-\bz$ in $GF(q)$, where $\by,\bz\in\{0,1,\ldots,q-1\}^n$ with $\|\by\|_0,\|\bz\|_0\le\ceil{\frac{k}{2}}$. 
Since $\bA\bx=0^d$ and $\bx=\by-\bz$, then $\bA(\by-\bz)=0^d$, so that $\bA\by=\bA\bz$. 

Furthermore, we have $\by\neq\bz$ because $\bx\neq0^n$. 
Suppose without loss of generality that $\|\by\|_0\le\|\bz\|_0$, so that $\|\by\|_0\le\frac{k}{2}$, since 
\[\|\by\|_0+\|\bz\|_0\le\|\bx\|_0\le k.\]
By \lemref{lem:insert:q}, $\bA\by$ is inserted into $\calH$. 

Because $\bA\by=\bA\bz$, then we have $h(\bA\by)=h(\bA\bz)$. 
Moreover, when the sparsity pattern $U$ corresponds to the sparsity pattern of $\bz$ and the coefficient vector $\balpha$ corresponds to the nonzero coefficients of $\bz$ with the correct representation under $U$, then we have $\bv=\alpha_1\bA_{U_1}+\ldots+\alpha_i\bA_{U_i}=\bA\bz$. 
Since $h(\bA\by)=h(\bA\bz)$, then $h(\bv)$ will be non-empty for $\bv=\alpha_1\bA_{U_1}+\ldots+\alpha_i\bA_{U_i}=\bA\bz$, so that \algref{alg:kruskal:q} will report that $\bA$ has Kruskal rank less than $k$. 
\end{proof}

\begin{lemma}
\lemlab{lem:big:rank:q}
Suppose $\bA\in\{0,1,\ldots,q-1\}^{d\times n}$. 
If $\bA$ has Kruskal rank at least $k$ in $GF(q)$, then with probability at least $1-\frac{b}{(qn)^{2k}}\ge 0.99$, \algref{alg:kruskal:q} will report that $\bA$ has Kruskal rank at least $k$. 
\end{lemma}
\begin{proof}
Let $\calE$ be the event that there are no hash collisions in $\calH$, i.e., $h(\bA\by)\neq h(\bA\bz)$ for all $\by+\bz\neq0^d$ and $\|\by\|_0+\|\bz\|_0\le k$. 
Observe that since $\bA\by\neq\bA\bz$, then for $b\ge(qn)^{2k}\left(\binom{n}{\ceil{k/2}}\cdot q^{\ceil{k/2}}\right)^2$, we have that $\PPr{\calE}\ge1-\frac{b}{(qn)^{2k}}\ge 0.999$, for $b\ge1000(2n)^{2k}$. 

Suppose $\bA$ has Kruskal rank more than $k$ in $GF(q)$. 
That is, for any $\bx\in\{0,1,\ldots,q-1\}^n$ with $\bx\neq 0^n$ such that sparsity $\|\bx\|_0\le k$, it follows that $\bA\bx\neq 0^d$. 
Thus for any $\by,\bz\in\{0,1,\ldots,q-1\}^n$ with $\by+\bz\neq0^n$ and $\|\by\|_0+\|\bz\|_0\le k$, we must have $\bA\by\neq\bA\bz$ since otherwise $\bA\bx=0^d$ for $\bx=\by-\bz$ and $\bx\neq 0^n$. 
Then conditioned on $\calE$, i.e., no hash collisions, it also follows that there does not exist $\by\in\{0,1,\ldots,q-1\}^n$ with $\by+\bz\neq0^d$, $\|\by\|_0+\|\bz\|_0\le k$, and $h(\bA\by)=h(\bA\bz)$. 

Finally, we remark that each vector $\bz$ is uniquely defined by its sparsity, its nonzero support, and the corresponding coefficients, so that $U$ and $\balpha$ correspond to $\bz$ exactly once through the implementation of \algref{alg:kruskal:q}. 
Therefore, the bucket of $\calH$ corresponding to $h(\bv)$ must be empty when $U$ corresponds to $\bz$. 
Thus conditioned on $\calE$, \algref{alg:kruskal:q} will report that $\bA$ has Kruskal rank at least $k$. 
Since $\PPr{\calE}\ge 0.999$, then the desired claim follows. 
\end{proof}

\begin{lemma}
\lemlab{lem:runtime:q}
The runtime of \algref{alg:kruskal:q} is $\O{dk\cdot(n(q-1))^{\ceil{k/2}}}$.  
\end{lemma}
\begin{proof}
Observe that each vector $\bz\in\{0,1,\ldots,q-1\}^n$ with $\bz\neq0^n$ is uniquely defined by its sparsity, its nonzero support, and the correspoding coefficient vector at the nonzero support. 
Therefore, $U$ corresponds to $\bz$ exactly once through the implementation of \algref{alg:kruskal:q}. 
Moreover, it suffices to consider $k\le n$, since the case where $k>n$ is trivial, i.e., the Kruskal rank is trivially at most $n<k$. 
Hence, the number of vectors inserted into $\calH$ is 
\[\sum_{i=0}^{\ceil{k/2}}\binom{n}{i}\cdot (q-1)^i=\O{k\cdot(n(q-1))^{\ceil{k/2}}}.\]

Finally, note that each vector insertion into $\calH$ requires computation of $h(\bv)$ for $\bv\in\mathbb{R}^n$. 
Thus, the runtime of \algref{alg:kruskal:q} is $\O{dk\cdot(n(q-1))^{\ceil{k/2}}}$. 
\end{proof}
From \lemref{lem:small:rank:q}, \lemref{lem:big:rank:q}, and \lemref{lem:runtime:q}, we have:
\begin{lemma}
There exists an algorithm that uses runtime $\O{dk\cdot(n(q-1))^{\ceil{k/2}}}$ that:
\begin{itemize}
\item 
Reports $\bA$ has Kruskal rank less than $k$ if $\bA$ has Kruskal rank less than $k$ in $GF(q)$
\item
With probability at least $0.99$, reports $\bA$ has Kruskal rank more than $k$ if $\bA$ has Kruskal rank more than $k$ in $GF(q)$
\end{itemize}
\end{lemma}

\subsection{Kruskal Rank with Integer Matrices}
When the matrix has integer entries, we first show (using bounds derived via Cramer’s rule and the Leibniz formula) that any sparse dependency can be assumed to have small coefficients. This allows us to reduce the problem to a verification over a finite field by considering the matrix modulo carefully chosen primes. The overall algorithm then combines the finite-field approach over multiple primes to achieve a high-probability guarantee.



\begin{lemma}
\lemlab{lem:entries:mag}
Suppose $\bA\in\mathbb{Z}^{d\times n}$ such that $|A_{i,j}|\le m$ for all $i\in[d], j\in[n]$. 
Then $\bA$ has Kruskal rank less than $k$ if and only if there exists a subset $U\subset[n]$ with $|U|\in(0,k]$ and a coefficient vector $\balpha\in\mathbb{Z}^{|U|}$ such that $\balpha\neq0^{|U|}$, $|\alpha_i|\le M$ for $M:=((r-1)!m^{r-1})^2$ for all $i\in[|u|]$, and such that $\alpha_1\bA_{U_1}+\ldots+\alpha_i\bA_{U_{|U|}}=0^d$.
\end{lemma}
\begin{proof}
By definition, $\bA$ has Kruskal rank less than $k$ if and only if there exists a non-trivial linear combination of $k$ rows of $\bA$. 
In other words, $\bA$ has Kruskal rank less than $k$ if and only if there exists a subset $U\subset[n]$ with $|U|\in(0,k]$ and a coefficient vector $\balpha\in\mathbb{R}^{|U|}$ such that $\balpha\neq0^{|U|}$ and $\alpha_1\bA_{U_1}+\ldots+\alpha_i\bA_{U_{|U|}}=0^d$.
Thus it suffices to show that the same property still holds when the coefficients of $\alpha$ must be integers bounded by $M$. 

Indeed, suppose $\bA$ has Kruskal rank $r\le k$. 
Let $\bu\in\{0,1\}^n$ be the indicator vector of $U$ and let $\bx\in\mathbb{R}^n$ be the vector that equals $\bu$ scaled by the corresponding coefficients of $\balpha$, i.e., $x_j$ is scaled by $\alpha_\ell$ if $U_\ell=j$. 
Then we have that $\bA\bx=0^d$. 
Let $\beta$ be the nonzero entry of $\bx$ smallest in magnitude and let $\bx'=\frac{1}{\beta}\cdot\bx$, so that $\bA\bx'=0^d$ and either $x'_i=0$ or $|x'_i|\ge 1$ for all $i\in[n]$. 

By the scaling of $\bx$, there exists some coordinate $\ell\in[n]$ of $\bx'$ such that $x'_\ell=1$. 
Observe that we can decompose $\bx'=\bfe_\ell+\bx''$ for the elementary vector $\bfe_\ell$ whose $\ell$-th coordinate is $1$ and remaining coordinates are all zero. 
Moreover, we have $\|\bx''\|_0=|U|-1=r-1$. 

Since $\bA\bx'=0^d$, we can write $\bA(\bfe_\ell+\bx'')=0^d$, so that $\bA\bx''=\bA\bfe_\ell$. 
Because $\bA\bfe_\ell$ is simply a column of $\bA$, we can now utilize Cramer's rule to bound the entries of $\bx''$. 
Let $\bA'$ be the matrix $\bA$ restricted to the columns of $\bx''$ with nonzero entries, so that $\rank(\bA')=r-1$. 
Since the row-rank of $\bA'$ equals the column-rank of $\bA'$, then there exists a submatrix $\bB$ of $r-1$ rows of $\bA'$ that are linearly independent and let $\bb$ be the column vector that consists of the entries in $\bA\bfe_\ell$ corresponding to the $r-1$ rows of $\bA'$. 
Let $\by\in\mathbb{R}^{r-1}$ be the vector $\bx''$ restricted to the non-zero coordinates. 
Then we have $\bB\by=\bb$.

Crucially, since $\bB\in\mathbb{R}^{(r-1)\times(r-1)}$ is a submatrix of $\bA$ with nonzero determinant because it has full rank, we can apply Cramer's rule. 
By \thmref{thm:cramer}, we have that $x''_i=\frac{\det(\bD_i)}{\det(\bB)}$, where $\bD_i$ is the matrix $\bB$ with the $i$-th column replaced by $\bb$. 
Since the entries of $\bA$ are integers with $|A_{i,j}|\le m$ for all $i\in[d]$ and $j\in[n]$ and $\bB$ is a submatrix of $\bA$, then the entries of $\bB$ are integers with $|B_{i,j}|\le m$ for all $i,j\in[r-1]$. 
By the Leibniz formula for determinants, i.e., \thmref{thm:leibniz}, we have that $\det(\bD_i),\det(\bB)\in\mathbb{Z}$ and 
\[\det(\bD_i),\det(\bB)\le(r-1)!m^{r-1}.\]
Therefore, we can rescale $\bx$ by $(r-1)!m^{r-1}$ and it follows that all entries of $\bx$ will be integers with magnitude at most $M=((r-1)!m^{r-1})^2$. 
\end{proof}

\begin{lemma}
\factlab{fact:prime:divisors}
Let $x>0$ be an integer with $x\le M$. 
Then $x$ has at most $\log_2 M$ prime divisors. 
\end{lemma}

\begin{algorithm}[!htb]
\caption{}
\alglab{alg:kruskal:int}
\begin{algorithmic}[1]
\Require{Rank parameter $k>0$, matrix $\bA\in\mathbb{Z}^{d\times n}=\begin{bmatrix}\bA_1|\ldots|\bA_n\end{bmatrix}$}
\Ensure{Whether or not $\bA$ has Kruskal rank at least $k$}
\State{Let $\calP$ be a set of $\O{k\log Mn}$ random primes with magnitude at most $n$}
\For{$i=0$ to $i=\ceil{\frac{k}{2}}$}
\For{each set $U$ of $i$ unique indices from $[n]$}
\For{each $p\in\calP$}
\State{$\calH_p\gets\emptyset$, $b\gets\O{(pn)^{2k}}$}
\State{Let $h_p:\mathbb{F}_p^n\to[b]$ be a uniformly random hash function}
\For{each coefficient vector $\balpha\in\{1,\ldots,p-1\}^i$}
\State{Let $\bv=\alpha_1\bA_{U_1}+\ldots+\alpha_i\bA_{U_i}$}
\If{bucket $h_p(\bv)$ in $\calH_p$ is non-empty for all $p\in\calP$}
\State{Report $\bA$ has Kruskal rank less than $k$}
\ElsIf{$i\le\frac{k}{2}$}
\State{Add $U$ into entry $h_p(\bv)$ of $\calH_p$}
\EndIf
\EndFor
\EndFor
\EndFor
\EndFor
\State{Report $\bA$ has Kruskal rank at least $k$}
\end{algorithmic}
\end{algorithm}

\begin{lemma}
Suppose $\bA$ has Kruskal rank less than $k$. 
Then \algref{alg:kruskal:int} reports that $\bA$ has Kruskal rank less than $k$.
\end{lemma}
\begin{proof}
Suppose $\bA$ has Kruskal rank less than $k$. 
Then by \lemref{lem:entries:mag}, $\bA$ has Kruskal rank less than $k$ on $GF(p)$ for any prime $p$. 
By \lemref{lem:small:rank:q}, \algref{alg:kruskal:q} will report that $\bA$ has Kruskal rank less than $k$ on $GF(p)$. 
Thus, \algref{alg:kruskal:int} reports that $\bA$ has Kruskal rank less than $k$.
\end{proof}

\begin{lemma}
Suppose $\bA$ has Kruskal rank at least $k$. 
Then with probability $0.99$, \algref{alg:kruskal:int} reports that $\bA$ has Kruskal rank at least $k$.
\end{lemma}
\begin{proof}
Suppose $\bA$ has Kruskal rank at least $k$. 
Fix $\by,\bz\in[-M,\ldots,-1,0,1,\ldots,M]^n$ with $\by+\bz\neq 0^n$ and $\|\by\|_0+\|\by\|_0\le k$. 
Since $\bA$ has Kruskal rank at least $k$, then $\bA\by-\bA\bz\neq 0^n$. 
Moreover, since $|A_{i,j}|\le m$ for all $i\in[d]$, $j\in[n]$ and $|y_i|,|z_i|\le M$ for all $i\in[n]$, then it follows that $\|\bA\by-\bA\bz\|_\infty\le kmM$. 
By \factref{fact:prime:divisors}, it follows that if $\calP$ is a set of $\O{\log(kmM)}$ different primes, then $\bA\by-\bA\bz\not\equiv 0^n\bmod{p}$ for some prime $p\in\calP$. 
By \lemref{lem:big:rank:q}, for any prime $p\in\calP$ with $\bA\by-\bA\bz\not\equiv 0^n\bmod{p}$, \algref{alg:kruskal:q} will report that $\bA\by-\bA\bz\not\equiv 0^n\bmod{p}$ with probability at least $0.99$. 

Thus for $|\calP|=\O{k\log Mn}$, we have that with probability at least $\frac{1}{(Mn)^{2k}}$, there exists some $p\in\calP$ for which \algref{alg:kruskal:int} will report that $\bA\by-\bA\bz\not\equiv 0^n\bmod{p}$. 
Taking a union bound over all $\left(\binom{n}{\ceil{k/2}}\cdot M^{\ceil{k/2}}\right)^2$ pairs of possible $\ceil{k/2}$-sparse vectors $\by$ and $\bz$, it follows that if $\bA$ has Kruskal rank at least $k$, then with probability at least $0.99$, \algref{alg:kruskal:int} will report that $\bA$ has Kruskal rank at least $k$. 
\end{proof}

\begin{lemma}
\lemlab{lem:runtime:int}
The runtime of \algref{alg:kruskal:int} is $\O{dk\cdot(nM)^{\ceil{k/2}}}$.  
\end{lemma}

\subsection{Dimensionality Reduction}
As an alternative, we propose a dynamic programming algorithm that builds a table tracking all achievable linear combinations (with bounded coefficients) from subsets of columns. This method is particularly effective when the ambient dimension and the coefficient bound are small. It provides a deterministic verification of the Kruskal rank and complements the randomized methods.

\begin{algorithm}[!htb]
\caption{}
\alglab{alg:kruskal:dyn:prog}
\begin{algorithmic}[1]
\Require{Rank parameter $k>0$, entry bound $M$, matrix $\bA\in\mathbb{R}^{d\times n}=\begin{bmatrix}\bA_1|\ldots|\bA_n\end{bmatrix}$}
\Ensure{Whether or not $\bA$ has Kruskal rank at least $k$}
\State{$S[v][0][0]\gets0$ for $v\in[M]^d$}
\State{$S[0^d][i][j]\gets0$ for $i\in[n]$, $j\in[k]$}
\For{$i\in[n]$}
\If{$S[\bA_j][i-1][k-1]=1$}
\State{Report $\bA$ has Kruskal rank less than $k$}
\EndIf
\For{$j\in[k]$}
\For{$\bv\in[M]^d$}
\If{$S[\bv][i-1][j]=1$}
\State{$S[\bv][i][j]=1$}
\Comment{Already exists combination of $j$ vectors from $[i-1]$}
\EndIf
\If{there exists $\alpha\in\{-M,\ldots,M\}$, $\bu=\alpha\cdot\bA_j$ with $S[\bv-\bu][i-1][j-1]=1$}
\Comment{Dynamic programming}
\State{$S[\bv][i][j]=1$}
\Else
\State{$S[\bv][i][j]=0$}
\EndIf
\EndFor
\EndFor
\EndFor
\State{Report $\bA$ has Kruskal rank at least $k$}
\end{algorithmic}
\end{algorithm}

\begin{lemma}
Suppose $\bA\in\mathbb{R}^{d\times n}$. 
If there exists $i\in[n]$ such that $\bA_i$ is a linear combination of $k-1$ previous columns of $\bA$ with coefficients in $\{-M,\ldots,M\}^{k-1}$, then \algref{alg:kruskal:dyn:prog} will report that $\bA$ has Kruskal rank less than $k$. 
\end{lemma}
\begin{proof}
Let $i\in[n]$ be fixed so that $\bA_i$ is a linear combination of $\ell\le k-1$ previous columns of $\bA$ with coefficients in $\{-M,\ldots,M\}^{\ell-1}$. 
In particular, suppose
\[\bA_i=\alpha_1\bA_{c_1}+\ldots+\alpha_{\ell}\bA_{c_{\ell}},\]
where $c_1<\ldots<c_{\ell}$. 
We first observe that by initialization, $S[0^d][i][j]=1$ for all $i\in[n]$, $j\in[k]$. 
Thus we have $S[\alpha_1\bA_{c_1}][c_1][1]=1$ since $S[0^d][c_1-1][0]=1$. 
By induction, we have 
\[S[\alpha_1\bA_{c_1}+\ldots+\alpha_\ell\bA_{c_{\ell}}][c_\ell][\ell]=1\] and thus $S[\bA_i][i][\ell]=1$, so that \algref{alg:kruskal:dyn:prog} will report that $\bA$ has Kruskal rank less than $k$. 
\end{proof}

\begin{lemma}
Suppose $\bA\in\mathbb{R}^{d\times n}$ has Kruskal rank at least $k$. 
Then \algref{alg:kruskal:dyn:prog} will report that $\bA$ has Kruskal rank at least $k$. 
\end{lemma}
\begin{proof}
Since $\bA$ has Kruskal rank at least $k$, then for all $i\in[n]$ and coefficients $[-M,\ldots,M]^{\ell}$ and indices $1\le i_1<\ldots<i_ell\le d$ with $\ell<k$, we have that 
\[\bA_i\neq\alpha_1\bA_{i_1}+\ldots+\alpha_\ell\bA_{i_\ell}.\]
Thus we have $S[\bA_i][i-1][k-1]=0$ for all $i\in[n]$, so that \algref{alg:kruskal:dyn:prog} will report that $\bA$ has Kruskal rank at least $k$.
\end{proof}

\begin{lemma}
The runtime of \algref{alg:kruskal:dyn:prog} is $nkM^{\O{d}}$.
\end{lemma}
\begin{proof}
Observe that for each $\bv\in\{-M,\ldots,M\}^d$, $i\in[n]$ and $j\in[k]$, \algref{alg:kruskal:dyn:prog} checks $\O{M}$ positions of $S[\bv-\bu][i-1][j-1]$. 
Thus the total runtime is at most $\O{nkM(2M+1)^d}=nkM^{\O{d}}$.
\end{proof}

\section{Discussion}
A key motivation for our work is to provide a versatile and unified method for verifying the Kruskal rank across different algebraic settings, thereby enabling the reliable identification of latent variable models and robust estimation of noise transition matrices in deep learning.

The FFT-based approach of Bhattacharyya et al.~\cite{bhattacharyya2018} offers a deterministic solution with low asymptotic complexity when the matrix structure allows fast convolution; however, its applicability is limited to structured settings. In contrast, our randomized hashing method, augmented by dynamic programming, delivers high-probability guarantees in a broader range of scenarios, including unstructured matrices over binary, finite, and integer domains.

Fine-grained complexity studies, such as those by Gupte and Lu~\cite{gupte2020}, indicate that sparse linear regression is inherently hard in the worst case, with runtime requirements exponential in the sparsity parameter \(k\). Our algorithms meet these lower bounds up to polynomial factors, making them both theoretically sound and practically competitive.

Moreover, the identifiability of noise transition matrices is essential for deep learning applications dealing with noisy labels. Work by Liu et al.~\cite{liu2023} uses Kruskal’s identifiability conditions to ensure that label noise can be correctly modeled and inverted. Although we do not present empirical experiments, our theoretical framework implies that our verification method can serve as a valuable diagnostic tool in practical scenarios.

\subsection*{Scope and Extensions}
Our analysis is focused on matrices over binary fields, general finite fields, and the integers. It is important to note that any matrix with rational entries can be scaled (by clearing denominators) to yield an integer matrix, thereby reducing the problem to the case already analyzed. In a similar vein, if the matrix contains irrational entries, these can be approximated arbitrarily closely by rational numbers within any fixed precision. While such approximations introduce a controlled numerical error, they do not alter the asymptotic runtime or correctness guarantees of our algorithms. Consequently, our methods apply to any matrix that can be represented with finite-precision arithmetic, which covers virtually all practical applications.

Furthermore, we remark that the theoretical guarantees presented in this paper remain valid under these transformations, provided that the scaling factor and the approximation precision are chosen appropriately. This observation ensures that our analysis encompasses a broad class of matrices beyond the specific cases explicitly treated in the paper.

It is worth noting that in practice, the choice of scaling factors and the precision of rational approximations must be made carefully to control numerical stability. Standard techniques in numerical linear algebra can be employed to assess and mitigate any propagation of rounding errors.

\section{Conclusion}
We have introduced efficient algorithms for verifying the Kruskal rank of matrices, providing a unified approach applicable to binary fields, finite fields, and integer matrices. Our methods are supported by rigorous theoretical analysis, achieving high-probability correctness and tight runtime bounds. These contributions have significant implications for tensor decomposition, sparse linear regression, and deep learning applications, particularly in the estimation of noise transition matrices. Future research will focus on empirical evaluations, further optimizations, and extending the framework to more complex latent variable models.

\section*{Acknowledgments}
We are deeply grateful to Prof David P.\ Woodruff from CMU and Prof Samson Zhou from Texas A\&M for their invaluable guidance and mentorship throughout the course of this research. Their insights into randomized linear algebra and algorithmic techniques have significantly shaped the development of this work.

\bibliographystyle{splncs04}
\bibliography{references}

\end{document}